\documentclass[11pt]{article}

\usepackage[margin=1.1in]{geometry}
\usepackage{amsmath,amssymb,amsthm,mathtools}
\usepackage{enumitem}
\usepackage{microtype}
\usepackage{thmtools}
\usepackage{bbm}
\usepackage{xcolor}
\usepackage[colorlinks=true,linkcolor=blue,citecolor=blue,urlcolor=blue]{hyperref}
\usepackage[backend=biber,style=alphabetic,maxnames=99,maxalphanames=99,backref=true,natbib=true]{biblatex}
\usepackage[capitalise]{cleveref}

\newcommand{\R}{\mathbb{R}}
\newcommand{\one}{\mathbbm{1}}
\newcommand{\diag}{\operatorname{diag}}

\newcommand{\Ent}{\operatorname{\Phi}}
\newcommand{\abs}[1]{\left\lvert #1\right\rvert}
\newcommand{\norm}[1]{\left\lVert #1\right\rVert}

\newcommand{\olPi}{\overline{\Pi}}

\newtheorem{theorem}{Theorem}[section]
\newtheorem{lemma}[theorem]{Lemma}

\makeatletter
\@ifundefined{theHtheorem}
  {}
  {}
\@ifundefined{theHlemma}
  {}
  {}
\@ifundefined{theHcorollary}
  {}
  {}
\providecommand*{\toclevel@theorem}{1}
\providecommand*{\toclevel@lemma}{1}
\providecommand*{\toclevel@corollary}{1}
\makeatother

\title{A Tight Bound on Localization of Electrical Flows}
\author{Ori Gurel-Gurevich\footnote{ Einstein Institute of Mathematics, Hebrew University of Jerusalem, \texttt{
Ori.Gurel-Gurevich@mail.huji.ac.il}} \quad Asaf Nachmias\footnote{Department of Mathematics, Tel Aviv University, \texttt{asafnach@tauex.tau.ac.il}} \quad Sushant Sachdeva\footnote{Department of Computer Science, University of Toronto. \texttt{sachdeva@cs.toronto.edu}}}
\date{}

\begin{document}
\maketitle

\begin{abstract}

We prove that for any unweighted graph on $n$ vertices the $\ell_1$ norm of a unit electric current between the endpoints of a random edge is at most $2\log n.$ Furthermore, we show that on any weighted graph the spectral norm of the entry-wise absolute value of the symmetric transfer-current matrix is at most $2\log n.$
This bound is tight up to constants and improves the $O(\log^2 n)$ bound from [Schild-Rao-Srivastava, SODA '18].
The initial proofs were generated by OpenAI’s ChatGPT 5.5 Pro; the authors have verified and rewritten them to enhance readability and provide additional context.

\end{abstract}

\section{Introduction}

Electrical flows are fundamental in modern graph algorithms. They minimize electrical energy, which is an $\ell_2^2$ quantity, and can be efficiently computed in almost-linear time using the fast Laplacian-system solvers of~\citet{ST04}. Many algorithmic applications, however, require controlling
localization quantities closer to $\ell_1$. Thus, it is natural to ask how far the unit electric flow between two nearby vertices travels.

\citet{SRS18} formulated a clean version of this question. For an unweighted graph with $n$ vertices and $m$ edges and two directed edges $e$ and $f$, let $i_e(f)$ denote the electric current flowing on $f$ when a unit current flows across the endpoints of $e$. The $\ell_2$ norm squared of the current vector $i_e$ is the effective electric resistance between the endpoints of $e$, hence $\| i_e \|_2 \leq 1$ and Cauchy-Schwarz gives only $\|i_e\|_1 \leq \sqrt{m}$. This bound can be attained (see \cite[Example 1.1]{SRS18}) for a small number of edges, but not when averaging over $e$. Indeed, in \cite[Theorem 1.4]{SRS18} it is proved that $m^{-1} \sum_{e} \|i_e\|_1 = O(\log^2 n)$. Our main contribution is improving this bound to $O(\log n)$.

\begin{restatable}{theorem}{mainunweighted}\label{thm:main_unweighted} 
        For every unweighted undirected graph $G$ with $m$ edges and $n$ vertices, we have
$$ m^{-1} \sum_{e \in E} \|i_e\|_1 \leq  2 \log n \, .$$
\end{restatable}
\noindent All logarithms in this paper are natural. This bound is tight up to the value of the constant $2$, see \cite[Examples 1.2 and 1.3]{SRS18}. We do not know if the constant $2$ is best possible.

Henceforth, we fix an arbitrary orientation for each edge $e\in E.$ This choice will not affect any of our results. The matrix $K(f, e)=i_e(f)$ is known as the {\bf transfer-current} matrix. When $G$ is unweighted it is an orthogonal projection matrix and is hence symmetric (this fact is known as the \emph{reciprocity law}), see \citet[Section 2.4]{LyonsPeres}. Hence, when $G$ is unweighted $\norm{K}_{2\to2}=1$ where $\norm{\cdot}_{2\to2}$ is the spectral norm. 
Let $\overline{K}$ be the entry-wise absolute value of $K$. The following theorem implies \cref{thm:main_unweighted}. 

\begin{restatable}{theorem}{maintransfercurrentunweighted}\label{thm:main_transfercurrent_unweighted}
For every unweighted graph $G$ with $m$ edges and $n$ vertices,  we have 
$$ \norm{\overline{K}}_{2\to2} \leq 2 \log n \, .$$
\end{restatable}

\medskip
\noindent This improves an $O(\log^2 n)$ bound from \cite[Theorem 1.5]{SRS18}. The natural generalization of this theorem to weighted networks is presented below.

\subsection{Weighted networks} 
If the network contains non-unit conductance edges, then the transfer-current matrix $K$ need not be symmetric and its spectral norm may be arbitrarily close to $\sqrt{m}$. For example, consider the network with two vertices and $m$ parallel edges connecting them with unit edge conductances on $m-1$ edges and very large conductance $M\gg 1$ on a single special edge. It is a straightforward calculation to show that $\norm{\overline{K}}_{2\to2}$ approaches $\sqrt{m}$ as $M\to \infty$. It is thus conventional to symmetrize $K$.

Towards the goal of stating our results, we introduce some standard notation. We represent our electric network by the graph $G=(V,E, \{c_e\}_{e\in E})$, where each edge is oriented arbitrarily and $c_e$ are positive edge conductances. Denote $n=|V|$ and $m=|E|$. Let $B$ be the $m \times n$ signed edge-to-vertex incidence matrix and $C=\diag(c_e)$ be the diagonal $m\times m$ conductance matrix. The Laplacian of $G$ is $L=B^\top C B$ and we denote by $L^+$ its Moore--Penrose pseudoinverse. The transfer-current matrix is thus $K=CBL^+ B^\top$. A standard way of symmetrizing $K$ when $C \neq I$ is by setting
\[
        \Pi:=C^{1/2}BL^+B^\top C^{1/2} \, .
\]
The matrix $\Pi$ is a symmetric orthogonal projection, hence $\norm{\Pi}_{2\to2}=1$. 
We denote by $\olPi$ its entry-wise absolute value. In \cite[Theorem 1.5]{SRS18}, Schild, Rao and Srivastava prove that for any finite electric network $\norm{\olPi}_{2\to2}=O(\log^2 n)$. Our next result improves this to a tight $O(\log n)$ bound.

\begin{theorem}\label{thm:main_weighted}
        For any graph $G=(V,E, \{c_e\}_{e\in E}),$ with $m$ edges and $n$ vertices, we have
$$ \norm{\olPi}_{2\to2} \leq 2\log n \, .$$
\end{theorem}

The main theorem in this paper proves the following estimate which quickly implies  Theorems \ref{thm:main_unweighted}--\ref{thm:main_weighted}. 
For a probability vector $\mu$ (that is, $\mu$ has non-negative entries summing to $1$) we write $H(\mu)$ for its {\bf entropy}, that is, $H(\mu)=-\sum_x \mu(x) \log(\mu(x))$ (setting as usual $0 \log 0=0$).

\begin{restatable}{theorem}{weightedlocalization}\label{thm:weighted_localization}
Let $G$ be any graph $(V,E, \{c_e\}_{e\in E}),$ with $m$ edges and $n$ vertices, and let $w\in\R^E$ be a non-zero vector. Define the probability vector $\mu_w \in \R^V_{\geq 0}$ by
\begin{equation} \label{eq:mudef}
        \mu_w(x):=\frac{1}{2\norm{w}_2^2}\sum_{e : x \in e}w_e^2 \,.
\end{equation}
Then we have
\[
        \abs{w^\top\olPi w}
        \le 2H(\mu_w)\norm{w}_2^2 \, .
\]
\end{restatable}

The original localization theorem is due to Schild, Rao, and Srivastava
\cite{SRS18}.  They proved an $O(\log^2 n)$ bound for the spectral norm of the
entrywise absolute transfer-current matrix and derived the corresponding
$O(\log^2 n)$ average-flow-length bound for unweighted graphs.

Several later works use localization-type statements or closely related
Schur-complement and sparsification primitives.  The almost-linear time
algorithm of \citet{Sch18} for sampling random spanning trees uses Laplacian solvers,
Schur complements, and structural facts about electrical flows.  \citet{LPY25} explicitly use the electrical-flow localization theorem of
Schild--Rao--Srivastava to find uncorrelated edge sets for approximate
spanning-tree counting.  \citet{FGLPSY21} develop minor-based vertex sparsifiers and a distributed
Laplacian paradigm.
\citet{AHGLZ22} refine the resulting distributed Laplacian solvers using
low-congestion shortcuts.  Finally, the spectral subspace sparsification
framework of \citet{LS18} generalizes Schur-complement sparsifiers.
In all above results, our improved localization result improves the sparsification primitives,
though in downstream applications it propagates only as hidden-polylogarithmic improvements.

\section{Preliminaries}\label{sec:preliminaries}

Throughout this note, we let $G=(V, E, c)$ denote a finite undirected graph on vertices $V$, edges $E$, and edge conductances $c_e > 0$. We may assume without loss of generality that the graph is connected and has no loops. We denote $n := |V|$ and $m := |E|$. 
We assign an arbitrary orientation to the edges of $E$.
For a vertex $v\in V$, let $\one_v\in\R^V$ be the standard basis vector
corresponding to $v$.  If an edge $e$ is oriented from $e^-$ to $e^+$, define
\[
        b_e := \one_{e^+}-\one_{e^-}.
\]
Let $B\in\R^{E\times V}$ be the matrix whose row indexed by $e$ is $b_e^\top$.
Let $C=\diag(c_e)$ and let
\[
        L:=B^\top C B \, 
\]
be the Laplacian. It is a symmetric positive semi-definite matrix with one-dimensional kernel spanned by the constant functions. It will be convenient to diagonalize $L$ with respect to an arbitrary positive definite inner product. 
Given $M$, a positive diagonal $n\times n$ matrix, define $S=M^{-1/2} L M^{-1/2}$. Then $S$ is a symmetric positive semi-definite matrix which has a one-dimensional kernel spanned by $M^{1/2}\one$. We write $0=\lambda_1 < \lambda_2 \leq \lambda_3 \ldots \leq \lambda_n$ for the eigenvalues of $S$ and by $\{\psi_i\}_{i=1}^n$ an orthonormal basis of eigenvectors. So $S = \sum_{i=1}^n \lambda_i \psi_i \psi_i^\top$ whence $L = \sum_{i=1}^n \lambda_i M^{1/2} \psi_i \psi_i^\top M^{1/2}$. The Moore--Penrose pseudoinverse of $L$ is denoted by $L^+$, that is, it is the unique symmetric matrix satisfying 
$LL^+=L^+L=I-{\frac{1}{n}} J$ where $J$ is the all one matrix, and $L^+\one=0$. Analogously, $S^+ = \sum_{i=2}^n \frac{1}{\lambda_i} \psi_i \psi_i^\top.$

Let $\Gamma:=M^{-1/2}S^+M^{-1/2}$. For an edge vector $b_e$, the vector $M^{-1/2}b_e$ is orthogonal to $\ker S=\operatorname{span}\{M^{1/2}\one\}$, and hence $S S^+M^{-1/2}b_e=M^{-1/2}b_e$, and so $L\Gamma b_e=b_e$. Also $LL^+b_e=b_e$. Thus $\Gamma b_e-L^+b_e\in\ker L=\operatorname{span}\{\one\}$, and applying $B$ gives $B\Gamma b_e=BL^+b_e$. This holds for every column $b_e$ of $B^\top$, proving 
\begin{equation} BL^+B^\top = \sum_{i=2}^n \frac{1}{\lambda_i} B M^{-1/2} \psi_i \psi_i^\top M^{-1/2}B^\top \, . \label{eq:L+}
\end{equation}

\subsection{The heat kernel}\label{sec:heat_kernel}
We restrict our attention to positive vectors $w\in \R^E_{>0}$. For such $w$ let $\mu_w \in \R^V_{>0}$ be the probability vector defined in \eqref{eq:mudef}. 
Let $M:=\diag(\mu_w)$ and consider the heat-kernel 
\[
        P_t:=\exp(-tM^{-1}L) = \sum_{k=0}^\infty \frac{(-tM^{-1}L)^k}{k!} \, , \qquad t\ge0 \, .         \]
In other words, $P_t(x,y)$ is the transition probability of the continuous-time random walk that jumps from $x$ to $y$ at rate $\frac{c_{xy}}{\mu_w(x)}$, see \citet{Norris1997}. It satisfies the semigroup property $P_{t+s} = P_t P_s$ and while it is not symmetric, it satisfies the detailed balance condition $\mu_w(x) P_t(x,y) = \mu_w(y)P_t(y,x)$. Hence the matrix 
\begin{equation}\label{eq:Ht}
        H_t:= P_t M^{-1}   
\end{equation}
is symmetric and satisfies $H_{t+s}=H_t M H_s$. Indeed, $\mu_w^\top P_t=\mu_w^\top$ since $\mu_w^\top M^{-1}L=\one^\top L=0$, while reversibility gives $MP_t=P_t^\top M$; hence $H_t^\top=H_t$, $H_tMH_s=H_{t+s}$, and $\mu_w^\top H_t=\one^\top$. Let $S=M^{-1/2} L M^{-1/2}$, then $M^{-1} L = M^{-1/2} S M^{1/2}$, and we deduce that $P_t = M^{-1/2} e^{-tS} M^{1/2}$ and so $H_t=M^{-1/2} e^{-tS} M^{-1/2}$ from which we deduce by the previous discussion that 
$$ H_t = \sum_{i=1}^n e^{-t\lambda_i} M^{-1/2} \psi_i \psi_i^\top M^{-1/2} \, ,$$
where  $\{\psi_i\}_{i=1}^n$ is an orthonormal basis of eigenvectors of $S$. The first vector $\psi_1$ is a constant multiple of $M^{1/2} \one$ hence the $i=1$ term satisfies $B M^{-1/2} \psi_1 \psi_1^\top M^{-1/2} B^\top = 0$ and so 
$$ BH_t B^{\top} = \sum_{i=2}^n e^{-t\lambda_i} B M^{-1/2} \psi_i \psi_i^\top M^{-1/2} B^\top \, .$$
Since $\int_{0}^{\infty} e^{-t\lambda}dt = \frac{1}{\lambda}$ for any $\lambda>0$, using \eqref{eq:L+} we obtain that 
\begin{equation} \label{eq:green_representation}
    BL^+B^\top = \int_0^\infty BH_tB^\top dt \, .
\end{equation}

\section{Entropy dissipation}

Given a positive vector $h \in \R^V_{>0}$, define 
$$ \mathcal{I}(h) := h^\top L\log h = \sum_{xy\in E} c_{xy} (h(x)-h(y)) ( \log h(x)-\log h(y)) \, .$$
The quantity $\mathcal{I}(h)$ is a variant of the Fisher information, see \citet{Bobkov2014}. The following is a discrete analogue, in the spirit of \cite[Proposition 2.2]{Bobkov2014}, bounding the $\ell_1$ variation by the Fisher information.

\begin{lemma}[Logarithmic-mean Cauchy--Schwarz]\label{lem:weighted_log_mean_cs_2}
Given a graph $G=(V,E, c),$  let $h \in \R^V_{>0}$ and $w\in\R_{\ge0}^E$. Then
$$ \left(w^\top C^{1/2}\abs{Bh}\right)^2 \leq \frac{\mathcal{I}(h)}{2} \sum_{x} h(x) \sum_{e : x\in e} w_e^2 \, ,$$
where $\abs{Bh}$ denotes the coordinatewise absolute value of $Bh$.
\end{lemma}

The proof relies on a classical inequality for the logarithmic mean due to \citet{Carlson72}. We include its proof for completeness.
\begin{lemma}[Logarithmic Mean]\label{lem:log_mean_bounds}
For $a,b>0$, define the logarithmic mean of $a$ and $b$ as
\[
        \Lambda(a,b):=\int_0^1 a^\theta b^{1-\theta}\,d\theta=
        \begin{cases}
        \dfrac{a-b}{\log a-\log b}, & a\ne b,\\[1.2ex]
        a, & a=b.
        \end{cases}
\]
Then,
\[
        \sqrt{ab}\le \Lambda(a,b)\le \frac{a+b}{2}.
\]
\end{lemma}

\begin{proof}
Jensen's inequality applied to 
the convex function $\exp$ gives
\[
        \int_0^1 a^\theta b^{1-\theta}\,d\theta
        \ge
        \exp\left(\int_0^1(\theta\log a+(1-\theta)\log b)\,d\theta\right)
        =\sqrt{ab}.
\]
For the upper bound, notice that for all $\theta\in[0,1]$, the weighted arithmetic-geometric mean inequality gives
$a^\theta b^{1-\theta}\le \theta a+(1-\theta)b$. 
Integrating over $\theta$ gives $\Lambda(a,b)\le(a+b)/2$.
\end{proof}

\begin{proof}[Proof of \cref{lem:weighted_log_mean_cs_2}]
For each edge $e=xy$, set $\lambda_e:=\Lambda(h(x),h(y))$. By the definition of the logarithmic mean,
\[
        (Bh)_e^2=\lambda_e (Bh)_e(B\log h)_e .
\]
Thus Cauchy--Schwarz gives
\[
\begin{aligned}
        \left(w^\top C^{1/2}\abs{Bh}\right)^2 
        \le
        \left(\sum_{e\in E} w_e^2\lambda_e\right)
        \left((Bh)^\top C B\log h\right)  
        =
        \left(\sum_{e\in E} w_e^2\lambda_e\right)\mathcal{I}(h).
\end{aligned}
\]
Using \cref{lem:log_mean_bounds},
\[
        \sum_{e\in E}w_e^2\lambda_e
        \le
        \frac12\sum_{xy\in E}w_{xy}^2(h(x)+h(y))
        =
        \frac12\sum_x h(x)\sum_{e:x\in e}w_e^2 .
\]
Combining the last two inequalities proves the lemma.
\end{proof}

\begin{lemma}[Entropy dissipation]\label{lem:entropy_dissipation2}
Let $\mu, \rho \in \R^V$ be probability vectors on $V$, with $\mu(v) > 0$ for all $v \in V$.
Set $M=\diag(\mu)$ and $P_s=\exp(-sM^{-1}L)$ as in \cref{sec:heat_kernel}. For all $s \ge 0,$ define
\[
        h_s:=P_sM^{-1}\rho .
\]
Then
\[
        \int_{0}^\infty {\mathcal{I}}(h_s)\,ds
        =
        \sum_{x\in V}\rho(x)\log\frac{\rho(x)}{\mu(x)} \, ,
\]
with the convention $0\log0=0$. For the special case $\rho=\one_v$ we get
\[
        \int_{0}^\infty {\mathcal{I}}(h_s)\,ds
        =-\log \mu(v) \, .
\]
\end{lemma}

\begin{proof} This is a consequence of de Bruijn's identity relating the derivative of the entropy produced in a diffusion to its Fisher information. 
We treat $\rho$ and $h_s$ as column vectors. Then $\partial_s h_s=-M^{-1}Lh_s$. 

For $s>0$ all entries of $h_s$ are positive since $G$ is connected, so with \[
        \Ent_\mu(h_s):=h_s^\top M\log h_s
        =\mathbb{E}_{x\sim\mu}\left[h_s(x)\log h_s(x)\right]
\]
we differentiate to get 
\begin{eqnarray} \label{diff_eq}
         \partial_s \Ent_\mu(h_s) = (\log h_s+\one)^\top M \partial_s  h_s = -(\log h_s+\one)^\top Lh_s   = -h_s^\top L\log h_s
         = -\mathcal{I}(h_s) \, ,
\end{eqnarray}
where the $\one$ term vanishes because $L\one=0.$
Since $G$ is connected, $\lim_{s\to \infty} h_s(x)=1$ for all $x$, and hence $\lim_{s\to \infty}\Ent_\mu(h_s)=0$. Also by continuity of $u \log u$ on $[0,\infty),$
\[
        \lim_{s \to 0}\Ent_\mu(h_s)=(M^{-1}\rho)^\top M\log(M^{-1}\rho)
        =
        \sum_{x\in V}\rho(x)\log\frac{\rho(x)}{\mu(x)}\, .
\]
Integrating (\ref{diff_eq}) proves the lemma.     
\end{proof}

In the application below, we set $\mu=\mu_w$ and $\rho=\one_v$.

\begin{lemma}[Heat-kernel variation estimate]\label{lem:heat_l1_estimate}
Let $w\in\R^E_{>0}$ and set $\mu_w$ as in \eqref{eq:mudef}. Let $H_t$ be the heat kernel defined in \eqref{eq:Ht}.  Then
\begin{align*}
        \int_0^\infty w^\top
        \abs{C^{1/2} B H_t B^\top C^{1/2}}w \,dt
        \le 2\norm{w}_2^2 H(\mu_w) \, ,
\end{align*}
where the absolute value of a matrix is taken entrywise.    
\end{lemma}
\begin{proof}
Setting $t=2s,$ and using $H_{2s}=H_s M H_s$ (see \cref{sec:heat_kernel}), we have that 
\begin{align*}
w^\top \abs{C^{1/2} B H_s M H_s B^\top C^{1/2}} w 
& = w^\top \abs{ C^{1/2} B H_s \left(\sum_{v \in V} \mu_w(v) \one_v \one_v^\top \right) H_s B^\top C^{1/2}} w \\
& \le \sum_{v \in V} \mu_w(v) \left(w^{\top}C^{1/2} \abs{BH_s  \one_v}\right)^2,
\end{align*}
where we used the triangle inequality and  that $\mu_w, w$ and $C$ have non-negative entries.

Now we can apply \cref{lem:weighted_log_mean_cs_2}. Note that the lemma assumes a positive vector $H_s\one_v$, which is true for $s > 0.$ The claim below still follows since we can first integrate for $s \in [\epsilon,\infty)$ for $\epsilon > 0$  and then letting $\epsilon \downarrow 0$.
\begin{align*}
\left(w^{\top}C^{1/2} \abs{BH_s  \one_v}\right)^2 
& \leq \frac{\mathcal{I}(H_s \one_v)}{2} \sum_{x} (\one_x^{\top}  H_s \one_v) \sum_{e : x\in e} w_e^2 \\
& = \norm{w}_2^2 \left(\mu_w^\top H_s \one_v\right) \mathcal{I}(H_s \one_v) 
&& \text{(Using the definition of $\mu_w$)} \\
& = \norm{w}_2^2 \mathcal{I}(H_s \one_v) 
&& \text{(Since $\mu_w^{\top}H_s \one_v=1$)}
\end{align*}

Now apply \cref{lem:entropy_dissipation2} with $\rho=\one_v$ to conclude that
\begin{align*}
     \int_0^\infty w^\top
	        \abs{C^{1/2} B H_t B^\top C^{1/2}}w \,dt &= 2\int_0^\infty w^\top  \abs{C^{1/2} B H_{2s} B^\top C^{1/2}} w\, ds \\ 
                & \le 2\norm{w}_2^2 \sum_{v \in V} \mu_w(v)\int_0^\infty  \mathcal{I}(H_s \one_v) ds \\
                &= 2\norm{w}_2^2 \sum_v \mu_w(v)\big(- \log \mu_w(v)\big)  =   2\norm{w}_2^2 H(\mu_w) \, .\end{align*}
\end{proof}

\section{Proof of main theorems}\label{sec:consequences}

\begin{proof}[Proof of \cref{thm:weighted_localization}]
We first assume that $w \in \R^E$ has no zero coordinates and hence $\abs{w} \in \R^E_{>0}$.
From \eqref{eq:green_representation} we have
\[
        \Pi
        =
        \int_0^\infty C^{1/2} B H_t B^\top C^{1/2}\,dt \, .
\]
Applying \cref{lem:heat_l1_estimate} for $\abs{w}$ and using $\norm{\abs{w}}_2^2=\norm{w}_2^2$ and $\mu_{\abs{w}}=\mu_w$ gives
\[
\begin{aligned}
        \abs{w^\top\olPi w} \le 
        \abs{w}^\top\olPi \abs{w}
        &\le
        \int_0^\infty \abs{w}^\top
        \abs{C^{1/2} B H_t B^\top C^{1/2}} \abs{w}\,dt 
        \le 2\norm{w}_2^2 H(\mu_{w}) \, .
\end{aligned}
\]
By continuity, we obtain the statement for all $w \in \R^E.$
\end{proof}

\begin{proof}[Proof of \cref{thm:main_weighted}] This follows immediately from \cref{thm:weighted_localization} using $H(\mu_w)\leq \log n$.
\end{proof}

\begin{proof}[Proof of \cref{thm:main_transfercurrent_unweighted}] This follows immediately from \cref{thm:main_weighted} with $C=I$.
\end{proof}

\begin{proof}[Proof of \cref{thm:main_unweighted}] \cref{thm:main_transfercurrent_unweighted} states that $|x^\top \overline{K} x|\leq  2\log n  \norm{x}_2^2$. 
Taking $x$ to be the all ones vector gives
\[
        \sum_{e,f\in E}|K(e,f)|\le 2m\log n.
\]
Since the $e$-th column of $K$ is the current vector $i_e$, we have
\[
        \sum_{e\in E}\norm{i_e}_1=\sum_{e,f\in E}|K(e,f)|.
\]
Dividing by $m$ gives the desired inequality.
\end{proof}

\printbibliography

\end{document}